\documentclass[leqno]{aadmbook}
\usepackage{amsthm}
\usepackage{amsfonts}
\usepackage{amsmath}
\newtheorem{theorem}{\bf Theorem}
\newtheorem{proposition}{Proposition}
\begin{document}


\oddsidemargin 16.5mm
\evensidemargin 16.5mm

\thispagestyle{plain}

\begin{center}
{\large \sc  Applicable Analysis and Discrete Mathematics}

{\small available online at  http:/$\!$/pefmath.etf.rs }
\end{center}

\noindent{\small{\sc  Appl. Anal. Discrete Math.\ }{\bf x} (xxxx),
xxx--xxx.} \hfill{\scriptsize doi:10.2298/AADMxxxxxxxx}

\vspace{5cc}
\begin{center}

{\large\bf  SIMPLE PARAMETRIZATION METHODS FOR GENERATING ADOMIAN POLYNOMIALS
\rule{0mm}{6mm}\renewcommand{\thefootnote}{}
\footnotetext{

\rule{2.4mm}{0mm}Keywords and Phrases. Adomian decomposition method, Adomian polynomials, formal power series, nonlinear operators.
}}

\vspace{1cc}
{\large\it K. K. Kataria, P. Vellaisamy}

\vspace{1cc}
\parbox{24cc}{{\small

In this paper, we discuss two simple parametrization methods for calculating Adomian polynomials for several nonlinear operators, which utilize the orthogonality of functions $e^{inx}$, where $n$ is an integer. Some important properties of Adomian polynomials are also discussed and illustrated with examples. These methods require minimum computation, are easy to implement, and are extended to multivariable case also. Examples of different forms of nonlinearity, which includes the one involved in the Navier Stokes equation, is considered. Explicit expression for the $n$-th order Adomian polynomials are obtained in most of the examples. 
%

}}
\end{center}

\vspace{1cc}


\vspace{1.5cc}
\begin{center}
{\bf 1. INTRODUCTION}
\end{center}

The Adomian decomposition method (ADM) (see {\small{\sc Adomian}} \cite{Adomian1986,Adomian1989,Adomian1994}) provides an analytical approximate solution for nonlinear functional equation in terms of a rapidly converging series, without linearization, perturbation or discretization. Consider a functional equation 
\begin{equation}\label{1.1}
u=f+L(u)+N(u),
\end{equation}
where $L$ and $N$ are respectively, linear and nonlinear operators and $f$ is a known function. In ADM, the solution $u(x,t)$ of (\ref{1.1}) is decomposed in the form of an infinite series,
\begin{equation*}\label{1.2}
u(x,t)= ‎‎\sum_{n=0}^{\infty‎}‎u_n(x,t).
\end{equation*}
Further, the nonlinear function $N(u)$ is assumed to admit the representation 
\begin{equation*}\label{1.3}
N(u)=‎‎\sum_{n=0}^{\infty‎}A_n(u_0,u_1,\ldots,u_n),
\end{equation*}
where $A_n$'s are $n$-th order Adomian polynomials. In the linear case, $N(u)=u$, $A_n$ simply reduces to $u_n$.

Adomian's method is simple in principle, but involves tedious calculations for obtaining Adomian polynomials. {\small{\sc Adomian}} \cite{Adomian1986} gave a method for determining these Adomian polynomials, by parametrizing $u(x,t)$ as
\begin{equation}\label{1.4}
u_\lambda(x,t)=‎‎\sum_{n=0}^{\infty‎}u_n(x,t)\lambda^n 
\end{equation}
and assuming $N(u_\lambda)$ to be analytic in $\lambda$, which decomposes as
\begin{equation}\label{1.5}
N(u_\lambda)=‎‎\sum_{n=0}^{\infty‎}A_n(u_0,u_1,\ldots,u_n)\lambda^n.
\end{equation}
Hence, the Adomian polynomials $A_n$ are given by
\begin{equation}\label{1.6}
A_n(u_0,u_1,\ldots,u_n)=\left.\frac{1}{n!}\frac{\partial^n N(u_\lambda)}{\partial \lambda^n} \right|_{\lambda=0},\ \forall\ n\in\mathbb{N}_0,
\end{equation}
where $\mathbb{N}_m=\left\{n\in\mathbb{N}\cup\{0\}:n\geq m\right\}$ and $\mathbb{N}$ denotes the set of positive integers.

{\small{\sc Rach}} \cite{Rach1984415} suggested the following formula for determining Adomian polynomials:
\begin{align}\label{1.7}
A_0(u_0)&=‎‎N(u_0),\nonumber\\
A_n(u_0,u_1,...,u_n)&=‎\sum_{k=1}^{n‎}C(k,n)N^{(k)}(u_0),\ \forall\ n\in\mathbb{N},
\end{align}
\noindent where 
\begin{equation*}\label{1.8}
C(k,n)=\underset{\sum_{j=1}^nk_j=k\ ,\ k_j\in\mathbb{N}_0}{\sum_{\sum_{j=1}^{n}jk_j=n}}\prod_{j=1}^{n}\frac{u_j^{k_j}}{k_j!}.
\end{equation*}

{\small{\sc Wazwaz}} \cite{Wazwaz200033} suggested a new algorithm in which after separating $A_0=N(u_0)$ from other terms of the Taylor series expansion of the nonlinear function $N(u)$, we collect all terms of the expansion obtained such that the sum of the subscripts of the components of $u(x,t)$ in each term is the same. The limitations of this algorithm is that it is difficult to keep track of the terms after some time. {\small{\sc Zhu}} et al. \cite{Zhu2005402} suggested another useful method, but it also involves tedious calculations of $n$-th derivative to obtain $A_n$. Adomian polynomials can also be obtained recursively (see {\small{\sc Biazar}} and {\small{\sc Shafiof}} \cite{Biazar2007975}, {\small{\sc Duan}} \cite{Duan20101235}, \cite{Duan20102456}, \cite{Duan20116337}). However, the disadvantage is that we do not have explicit form for $A_n$'s.

In this paper, we develop a simple parametrization technique for calculating Adomian polynomials and discuss some of their important properties. Indeed, we develop two new simple methods to generate Adomian polynomials using the orthogonality of functions $\{e^{inx}, n\in\mathbb{Z}\}$. The first method determines these polynomials explicitly, whereas the second method generates them recursively. The newly developed techniques are more viable, require less computation and generate Adomian polynomials in a fewer steps. Both the methods are extended to the case of several variables. Different forms of nonlinearity are discussed as applications of our methods.

\begin{center}
{\bf 2. ADOMIAN POLYNOMIALS AND PARAMETRIZATION METHODS}
\end{center}

We assume the following hypotheses (see {\small{\sc Cherruault}} and {\small{\sc Adomian}} \cite{Cherruault1993103}):
\begin{itemize}
\item[$H1:$] The series solution $u=\sum_{k=0}^{\infty}u_k$ of (\ref{1.1}) is absolutely convergent,
\item[$H2:$] The nonlinear function $N(u)$ admits the representation
\end{itemize}
\begin{equation}\label{2.1}
N(u)=\sum_{k=0}^{\infty}N^{(k)}(0)\frac{u^k}{k!}\ ,\ \ \ \ |u|<\infty.
\end{equation}

The assumption $H2$, is almost always satisfied in concrete physical problems. By $H1$ and $H2$, we have as a generalization of Taylor series, the Adomian series (see {\small{\sc Cherruault}} and {\small{\sc Adomian}} \cite{Cherruault1993103})
\begin{equation}\label{2.2}
N(u)=\sum_{k=0}^{\infty}A_k(u_0,u_1,\ldots,u_k)=\sum_{k=0}^{\infty}N^{(k)}(u_0)\frac{(u-u_0)^k}{k!}.
\end{equation}

Note that (\ref{2.2}) is a rearrangement of an absolutely convergent series (\ref{2.1}). We look at a more general form of parametrization than the one given in (\ref{1.4}). We consider the following parametrization of $u(x,t)$ and its complex conjugate $\overline{u}(x,t)$:
\begin{equation}\label{2.3}
u_\lambda(x,t)=‎‎\sum_{k=0}^{\infty‎}u_k(x,t)f^k(\lambda)\ \ \mathrm{and}\ \ \overline{u}_\lambda(x,t)=‎‎\sum_{k=0}^{\infty‎}\overline{u}_k(x,t)f^k(\lambda),
\end{equation}
\noindent where $\lambda$ is a real parameter and $f$ is any real or complex valued function with $|f|<1$.

Note that for such a parametrization, series (\ref{2.3}) is also absolutely convergent. Now using (\ref{2.2}) and (\ref{2.3}), we have
\begin{equation}\label{2.5}
N(u_\lambda)=\sum_{k=0}^{\infty}\frac{N^{(k)}(u_0)}{k!}\left(\sum_{j=1}^{\infty‎}u_j(x,t)f^j(\lambda)\right)^k.
\end{equation}

Since $\sum_{j=1}^{\infty‎}u_j(x,t)f^j(\lambda)$ is absolutely convergent, by a rearrangement of terms in the right hand side of (\ref{2.5}), we can write $N(u_\lambda)$ as $\sum_{k=0}^{\infty}A_kf^k(\lambda)$, where $A_k$'s are Adomian polynomials. Hence,
\begin{align}\label{2.6}
N(u_\lambda)&=N(u_0)+N^{(1)}(u_0)\left(u_1f(\lambda)+u_2f^2(\lambda)+\ldots\right)\nonumber\\
&\ \ \ +\frac{N^{(2)}(u_0)}{2!}\left(u_1f(\lambda)+u_2f^2(\lambda)+\ldots\right)^2\nonumber\\
&\ \ \ +\frac{N^{(3)}(u_0)}{3!}\left(u_1f(\lambda)+u_2f^2(\lambda)+\ldots\right)^3+\ldots\nonumber\\ 
&=N(u_0)+N^{(1)}(u_0)u_1f(\lambda)+\left(N^{(1)}(u_0)u_2+N^{(2)}(u_0)\frac{u_1^2}{2!}\right)f^2(\lambda)\nonumber\\
&\ \ \ +\left(N^{(1)}(u_0)u_3+N^{(2)}(u_0)u_1u_2+N^{(3)}(u_0)\frac{u_1^3}{3!}\right)f^3(\lambda)+\ldots \nonumber\\
&=\sum_{k=0}^{\infty}A_k(u_0,u_1,\ldots,u_k)f^k(\lambda).
\end{align}

Note that $A_k$'s are polynomials in $u_0,u_1,\ldots,u_k$ only. For a suitable choice of $f$, we possibly can develop a convenient method to determine these Adomian polynomials. One such method was given by Adomian himself where he chooses $f(\lambda)=\lambda$ and then taking $n$-th derivative on both sides of (\ref{2.6}) obtained (\ref{1.6}). In Section 4, we choose $f(\lambda)=e^{i\lambda}$ and develop two new methods to determine Adomian polynomials.

\begin{center}
{\bf 3. SOME PROPERTIES OF ADOMIAN POLYNOMIALS}
\end{center}

In this section, we discuss some important properties of Adomian polynomials, which are very useful in many cases to obtain them without explicit calculations. Indeed, a formal power series can be effectively used to obtain them.

Let $f$ and $g$ be formal power series in $x$ with $f(x)=\sum_{k=0}^{\infty}a_kx^k$ and $g(x)=\sum_{k=0}^{\infty}b_kx^k$. Then,
\begin{equation}\label{3.1}
\frac{g(x)}{f(x)}=\sum_{k=0}^{\infty}c_kx^k,\ \ c_0=\frac{b_0}{a_0},\ c_k=\frac{1}{a_0}\left(b_k-\sum_{j=1}^{k}a_jc_{k-j}\right),
\end{equation}
and for any $n\in\mathbb{N}$, we have
\begin{equation}\label{3.2}
f^n(x)=\sum_{k=0}^{\infty}d_kx^k,\ \ d_0=a_0^n,\ d_k=\frac{1}{ka_0}\sum_{j=1}^{k}(jn-k+j)a_jd_{k-j},
\end{equation}
provided $a_0$ is invertible.

Now we state and prove some general properties of Adomian polynomials.
\begin{theorem}\label{t1}
Let $A_{1_n},A_{2_n},\ldots,A_{m_n}$ be the Adomian polynomials of the nonlinear operators $N_1,N_2,\ldots,N_m$, respectively. Then the Adomian polynomials of\\
\noindent (i) $N(u)=\sum_{k=1}^{m}\alpha_k N_k(u)$ are given by $$A_n=\sum_{k=1}^{m}\alpha_k A_{k_n},\ \forall\ n\in\mathbb{N}_0$$ where the $\alpha_k$'s are scalars.\\
\noindent (ii) $N(u)=\prod_{k=1}^{m}N_k(u)$ are given by
\begin{equation}\label{3.3}
A_n=\underset{k_j\in\mathbb{N}_0}{\sum_{\sum_{j=1}^{m}k_j=n}}\prod_{j=1}^{m}A_{j_{k_j}},\ \forall\ n\in\mathbb{N}_0.
\end{equation}
In particular, Adomian polynomials of $N(u)=N_1(u)N_2(u)$, from (\ref{3.3}), are
\begin{equation*}
A_n=\sum_{k=0}^{n}A_{1_k}A_{2_{n-k}}.
\end{equation*} 
\noindent (iii) $N(u)=N_1(u)/N_2(u)$ are given by $A_0=A_{1_0}/A_{2_0}$ and
\begin{equation*}\label{3.4}
A_n=\frac{1}{A_{2_0}}\left(A_{1_n}-\sum_{k=1}^{n}A_{2_k}A_{n-k}\right),\ \forall\ n\in\mathbb{N}.
\end{equation*}
\noindent (iv) $N(u)=N^p_1(u)$ for any $p\in\mathbb{N}$ are given by $A_0=A^p_{1_0}$ and
\begin{equation*}\label{3.5}
A_n=\frac{1}{nA_{1_0}}\sum_{k=1}^{n}(kp-n+k)A_{1_k}A_{n-k},\ \forall\ n\in\mathbb{N}.
\end{equation*}
\noindent (v) $N(u)=N_1\left(N_2(u)\right)$ are given by $A_0=N_1\left(A_{2_0}\right)$ and
\begin{equation}\label{3.6}
A_n=\underset{k_j\in\mathbb{N}_0}{\sum_{\sum_{j=1}^{n}jk_j=n}}N_1^{(\sum_{j=1}^{n}k_j)}\left(A_{2_0}\right)\prod_{j=1}^{n}\frac{A_{2_j}^{k_j}}{k_j!},\ \forall\ n\in\mathbb{N}.
\end{equation}
\end{theorem}
\begin{proof}
\noindent (i) Directly follows from (\ref{1.6}).\\
\noindent (ii) Note that Leibniz rule (see {\small{\sc Johnson}} \cite{Johnson02thecurious}) for higher derivatives of product of $m$ functions is given by
\begin{equation}\label{3.7}
\frac{d^n}{dt^n}\left(f_1(t)f_2(t)\ldots f_m(t)\right)=\underset{k_j\in\mathbb{N}_0}{\sum_{\sum_{j=1}^{m}k_j=n}}n!\prod_{j=1}^{m}\frac{f_j^{(k_j)}(t)}{k_j!}.
\end{equation}
Using (\ref{1.6}) and (\ref{3.7}), the Adomian polynomials are
\begin{align*}
A_n(u_0,u_1,\ldots,u_n)&=\left.\frac{1}{n!}\frac{\partial^n N_1(u_\lambda)N_2(u_\lambda)\ldots N_m(u_\lambda)}{\partial \lambda^n} \right|_{\lambda=0}\\
&=\frac{1}{n!}\underset{k_j\in\mathbb{N}_0}{\sum_{\sum_{j=1}^{m}k_j=n}}n!\prod_{j=1}^{m}\left.\frac{1}{k_j!}\frac{\partial^{k_j}N_j(u_\lambda)}{\partial\lambda^{k_j}}\right|_{\lambda=0}\\
&=\underset{k_j\in\mathbb{N}_0}{\sum_{\sum_{j=1}^{m}k_j=n}}\prod_{j=1}^{m}A_{j_{k_j}},\ \forall\ n\in\mathbb{N}_0.
\end{align*}
\noindent (iii) Follows directly from (\ref{1.5}) and (\ref{3.1}), whereas (iv) follows from (\ref{1.5}) and (\ref{3.2}).\\
\noindent (v) By using Fa\`{a} di Bruno's formula (see {\small{\sc Johnson}} \cite{Johnson02thecurious}) for generalized chain rule for higher derivatives of composition of two functions,
\begin{equation*}\label{3.8}
\frac{d^n}{dt^n}g\big(f(t)\big)=\underset{k_j\in\mathbb{N}_0}{\sum_{\sum_{j=1}^{n}jk_j=n}}n!g^{(\sum_{j=1}^{n}k_j)}\big(f(t)\big)\prod_{j=1}^{n}\frac{1}{k_j!}\left(\frac{f^{(j)}(t)}{j!}\right)^{k_j},\ \forall\ n\in\mathbb{N},
\end{equation*}
we get from (\ref{1.6}),
\begin{align*}
A_n(u_0,u_1,\ldots,u_n)&=\left.\frac{1}{n!}\frac{\partial^n N_1\left(N_2(u_\lambda)\right)}{\partial \lambda^n} \right|_{\lambda=0}\\
&=\underset{k_j\in\mathbb{N}_0}{\sum_{\sum_{j=1}^{n}jk_j=n}}N_1^{(\sum_{j=1}^{n}k_j)}\left(N_2(u_\lambda)\right)\prod_{j=1}^{n}\left.\frac{1}{k_j!}\left(\frac{1}{j!}\frac{\partial^jN_2(u_\lambda)}{\partial \lambda^j}\right)^{k_j}\right|_{\lambda=0}\\
&=\underset{k_j\in\mathbb{N}_0}{\sum_{\sum_{j=1}^{n}jk_j=n}}N_1^{(\sum_{j=1}^{n}k_j)}\left(A_{2_0}\right)\prod_{j=1}^{n}\frac{A_{2_j}^{k_j}}{k_j!},\ \forall\ n\in\mathbb{N}.
\end{align*}
This completes the proof.
\end{proof}
\noindent {\small{\sc Remark 1.}} {\small{\sc Adomian}} and {\small{\sc Rach}} \cite{Adomian1986504} proposed an algorithm for obtaining Adomian polynomials of composite nonlinearity, whereas (\ref{3.6}) gives an explicit formula. Also, Rach formula (\ref{1.7}) is a particular case of (\ref{3.6}) for composed function $N(u_\lambda)$.

\begin{center}
{\bf 4. TWO SIMPLE METHODS TO CALCULATE ADOMIAN POLYNOMIALS}
\end{center}

In this section, we give two new methods to calculate Adomian polynomials. The basic idea is to avoid the tedious calculations of higher derivatives involved in the existing methods. Let $\mathbb{Z}$ denote the set of all integers. Consider the set of orthogonal functions $\{e^{inx}, n\in\mathbb{Z}\},$ which indeed forms a basis for the Hilbert space $L^2[-\pi,\pi]$ with inner product
\begin{equation*}\label{4.1}
<f ,g >=\int^\pi_{-\pi} f(x)\overline{g(x)} \,dx.
\end{equation*}
Specifically, we use the fact
\begin{equation}\label{4.2}
 <e^{in\lambda},e^{im\lambda} >=\int^\pi_{-\pi} e^{in\lambda} e^{-im\lambda} \,d\lambda =\left\{
	\begin{array}{ll}
	    0, & \mbox{if } m\neq n,\\
		2\pi,  & \mbox{if } m=n.
	\end{array}
\right.
\end{equation}
We choose $f(\lambda)=e^{i\lambda}$ in (\ref{2.3}), so that
\begin{equation}\label{4.3}
u_\lambda(x,t)=‎‎\sum_{k=0}^{\infty‎}u_k(x,t) e^{ik\lambda}
\end{equation}
and its complex conjugate, $\overline{u}(x,t)$ is parametrized as $$\overline{u}_\lambda(x,t)=‎‎\sum_{k=0}^{\infty‎}\overline{u}_k(x,t) e^{ik\lambda}.$$
\noindent {\small{\sc Remark 2.}}
Note that $u_\lambda$ in (\ref{4.3}), as a function of $\lambda$, is a series of periodic functions each of period $2\pi$ and therefore $N(u_\lambda)$ is also $2\pi$-periodic. The absolute convergence of $u_\lambda=‎‎\sum_{k=0}^{\infty‎}u_k e^{ik\lambda}$ and $N(u_\lambda)$ follow from hypotheses $H1$ and $H2$. Also, for parametrization (\ref{4.3}), Adomian polynomials for the nonlinear function $N(u)$ turn out to be the Fourier coefficients of the periodic function $N(u_\lambda)$.
\begin{theorem}\label{t2}
Let $u_\lambda=\sum_{k=0}^{\infty‎}u_k e^{ik\lambda}$ be a parametrized representation of $u(x,t)$, where $\lambda$ is a real parameter and $N$ be the nonlinear function defined in (\ref{1.1}). Then,
\begin{equation*}\label{4.4}
\int^\pi_{-\pi} N\left(u_\lambda\right) e^{-in\lambda} \,d\lambda  = \int^\pi_{-\pi} N\left(‎‎\sum_{k=0}^{n‎}u_k e^{ik\lambda}\right) e^{-in\lambda} \,d\lambda,\ \forall\ n\in\mathbb{N}_0.
\end{equation*}
\end{theorem}
\begin{proof}
From the assumption $H1$, $\sum_{j=1}^{\infty‎}|u_j|=M<\infty$. Therefore, from (\ref{2.2}), the $k$-th term in $N(u_{\lambda})$ is
\begin{eqnarray*}\label{4.5}
\left|\frac{N^{(k)}(u_0)}{k!}\left(\sum_{j=1}^{\infty‎}u_j e^{ij\lambda}\right)^k\right| &\leq&
\left|\frac{N^{(k)}(u_0)}{k!}\right|\left(\sum_{j=1}^{\infty‎}|u_j|\right)^k=\left|\frac{N^{(k)}(u_0)}{k!}\right|M^k.
\end{eqnarray*}
Since (\ref{2.2}) is an absolutely convergent series with infinite radius of convergence, $\sum_{k=0}^{\infty}\left|\frac{N^{(k)}(u_0)}{k!}\right|M^k$ converges. By Weierstrass M-test, the series
\begin{equation*}
\sum_{k=0}^{\infty}\frac{N^{(k)}(u_0)}{k!}\left(\sum_{j=1}^{\infty‎}u_j e^{ij\lambda}\right)^k
\end{equation*}
 converges uniformly. Hence, using (\ref{2.2}), we get for $n\in\mathbb{N}_0$
\begin{align*}\label{new}
\int^\pi_{-\pi} N(u_\lambda) e^{-in\lambda} \,d\lambda
&=\int^\pi_{-\pi}\sum_{k=0}^{\infty} \frac{N^{(k)}(u_0)}{k!}\left(‎‎\sum_{j=1}^{n}u_j e^{ij\lambda}+‎‎\sum_{j=n+1}^{\infty‎}u_j e^{ij\lambda}\right)^ke^{-in\lambda} \,d\lambda\nonumber\\
&=\underset{m\rightarrow \infty}{\lim}\sum_{k=0}^{m}\int^\pi_{-\pi} \frac{N^{(k)}(u_0)}{k!}\left(‎‎\sum_{j=1}^{n}u_j e^{ij\lambda}\right)^ke^{-in\lambda} \,d\lambda,\ \ \mathrm{(by\ (\ref{4.2}))}\\
&=\int^\pi_{-\pi}\sum_{k=0}^{\infty}\frac{N^{(k)}(u_0)}{k!}\left(‎‎\sum_{j=0}^{n}u_j e^{ij\lambda}-u_0\right)^ke^{-in\lambda} \,d\lambda\\
&=\int^\pi_{-\pi} N\left(‎‎\sum_{k=0}^{n‎}u_k e^{ik\lambda}\right) e^{-in\lambda} \,d\lambda,
\end{align*}
where the last step follows from (\ref{2.2}). This completes the proof.
\end{proof}

Using Theorem \ref{t2}, we propose two methods to calculate Adomian polynomials.

\subsection*{4.1. First Method}

\hspace{.75cm} Let $u_\lambda=‎‎\sum_{k=0}^{\infty‎}u_k e^{ik\lambda}$ and $N(u_\lambda)=‎‎\sum_{k=0}^{\infty‎}A_k e^{ik\lambda},$ where  $A_k$'s are Adomian polynomials. Then
\begin{equation}\label{4.6}
\int^\pi_{-\pi} N\left(‎‎\sum_{k=0}^{\infty‎}u_k e^{ik\lambda}\right) e^{-in\lambda} \,d\lambda =\int^\pi_{-\pi} ‎‎\sum_{k=0}^{\infty‎}A_k e^{ik\lambda} e^{-in\lambda}\,d\lambda=2\pi A_n.
\end{equation}
The last equality in (\ref{4.6}) follows due to the uniform convergence of the series $\sum_{k=0}^{\infty‎}A_k e^{i(k-n)\lambda}$. Hence,
\begin{align}
A_n(u_0,u_1,\ldots,u_n)&=\frac{1}{2\pi}\int^\pi_{-\pi} N\left(‎‎\sum_{k=0}^{\infty‎}u_k e^{ik\lambda}\right) e^{-in\lambda} \,d\lambda\nonumber\\
&=\frac{1}{2\pi}\int^\pi_{-\pi} N\left(‎‎\sum_{k=0}^{n‎}u_k e^{ik\lambda}\right) e^{-in\lambda} \,d\lambda,\ \forall\ n\in\mathbb{N}_0,\label{4.8}
\end{align}
by Theorem \ref{t2}.

\subsection*{4.2. Second Method}

\hspace{.75cm} {\small{\sc Biazar}} and {\small{\sc Shafiof}} \cite{Biazar2007975} proposed a recursive method to calculate Adomian polynomials, in which only one time differentiation is required. Some useful recursive relationships among the index vectors of the Adomian polynomials were obtained by {\small{\sc Duan}} \cite{Duan20101235}. These relationships within index vectors easily generates Adomian polynomials on using Rach formula (\ref{1.7}).

Here, we propose a new recursive method for calculating Adomian polynomials by using a different approach. Define an operator $T$ by
\begin{equation}\label{4.9}
T(A_{n}(u_0,u_1,\ldots,u_{n}))=\frac{1}{2\pi}\int^\pi_{-\pi}A_{n}(v_0,v_1,\ldots,v_{n})e^{-i\lambda} \,d\lambda,
\end{equation}
where\ $v_k=u_k+(k+1)u_{k+1}e^{i\lambda}$ and $\overline{v}_k=\overline{u}_k+(k+1)\overline{u}_{k+1}e^{i\lambda}$, $k\in\{0,1,2,\ldots,n\}$.
\begin{proposition}\label{l4.1}
Let $u=\sum_{k=0}^{\infty}u_k$ be the solution of (\ref{1.1}) and $N$ be a nonlinear operator. Then, operator $T$ given by (\ref{4.9}) satisfies the following properties.\\
\noindent (i) $T(u_k)=(k+1)u_{k+1},$\\
\noindent (ii) $T\left(N^{(k)}(u_0)\right)=u_1N^{(k+1)}(u_0),$\\
\noindent (iii) $T(u_{k_1}u_{k_2}\ldots u_{k_m})=u_{k_1}T(u_{k_2}u_{k_3}\ldots u_{k_m})+u_{k_2}u_{k_3}\ldots u_{k_m}T(u_{k_1}),$\\
\noindent (iv) $T\left(u_{k_1}u_{k_2}\ldots u_{k_m}N^{(k)}(u_0)\right)=u_{k_1}u_{k_2}\ldots u_{k_m}T\left(N^{(k)}(u_0)\right)\\~~~~~~~~~~~~~~~~~~~~~~~~~~~~~~~~~~~~~~~~~~+T(u_{k_1}u_{k_2}\ldots u_{k_m})N^{(k)}(u_0),$\\
\noindent (v) $T\left(\alpha u_{k_1}u_{k_2}\ldots u_{k_m}N^{(k)}(u_0)+\beta u_{j_1}u_{j_2}\ldots u_{j_l}N^{(k')}(u_0)\right)\\~~~~~~~~~~~~~~~~~~~~~~~~~~~~~~=\alpha T\left(u_{k_1}u_{k_2}\ldots u_{k_m}N^{(k)}(u_0)\right)+\beta T\left(u_{j_1}u_{j_2}\ldots u_{j_l}N^{(k')}(u_0)\right),$\\ where $k,k',k_i,j_i\in\mathbb{N}_0$; $m,l\in\mathbb{N}_2$ and $\alpha,\beta$ are scalars.
\end{proposition}
\begin{proof}
\noindent Parts (i), (iii) and (v) follow easily by using (\ref{4.2}).\\
\noindent (ii) From (\ref{4.9}), we have 
\begin{equation}\label{4.10}
T\left(N^{(k)}(u_0)\right)=\frac{1}{2\pi}\int^\pi_{-\pi}N^{(k)}\left(u_0+u_1e^{i\lambda}\right)e^{-i\lambda} \,d\lambda,\ \forall\ k\in\mathbb{N}_0.
\end{equation}
From (\ref{4.8}), the left hand side of (\ref{4.10}) is $A_1$ for $N^{(k)}(u)$, which by (\ref{1.7}) is equal to $u_1N^{(k+1)}(u_0)$. This can also be obtained directly on using (\ref{2.2}) and (\ref{4.2}).\\
\noindent (iv) Using (\ref{2.2}) and (\ref{4.2}), we get
\begin{equation}\label{4.11}
\frac{1}{2\pi}\int^\pi_{-\pi}N^{(k)}\left(u_0+u_1e^{i\lambda}\right)\,d\lambda=N^{(k)}\left(u_0\right).
\end{equation}
From (\ref{4.9}), we have
\begin{align*}
T\bigg(u_{k_1}u_{k_2}&\ldots u_{k_m}N^{(k)}(u_0)\bigg)\\&=\frac{1}{2\pi}\int^\pi_{-\pi}\prod_{j=1}^{m}\left(u_{k_j}+(k_j+1)u_{k_j+1}e^{i\lambda}\right)N^{(k)}\left(u_0+u_1e^{i\lambda}\right)e^{-i\lambda} \,d\lambda\\
&=\prod_{j=1}^{m}u_{k_j}\frac{1}{2\pi}\int^\pi_{-\pi}N^{(k)}\left(u_0+u_1e^{i\lambda}\right)e^{-i\lambda} \,d\lambda\\
&\ \ \ +\sum_{l=1}^{m}(k_l+1)u_{k_l+1}\underset{j\neq l}{\prod_{j=1}^{m}}u_{k_j}\frac{1}{2\pi}\int^\pi_{-\pi}N^{(k)}\left(u_0+u_1e^{i\lambda}\right) \,d\lambda,\ \ \mathrm{(by\ (\ref{4.2}))}\\
&=u_{k_1}u_{k_2}\ldots u_{k_m}T\left(N^{(k)}(u_0)\right)+T(u_{k_1}u_{k_2}\ldots u_{k_m})N^{(k)}(u_0),\ \forall\ m\geq2.
\end{align*}
The last equality follows from (\ref{4.10}) and (\ref{4.11}).
\end{proof}
For an operator $T$ satisfying the above properties, the following result due to {\small{\sc Babolian}} and {\small{\sc Javadi}} \cite{Babolian2004253} holds:
\begin{equation}\label{4.12}
A_n(u_0,u_1,\ldots,u_n)=\frac{1}{n}T\left(A_{n-1}(u_0,u_1,\ldots,u_{n-1})\right).
\end{equation}
After calculating $A_0$ from (\ref{4.8}) as
\begin{equation}\label{4.13}
A_0(u_0)=N(u_0),
\end{equation}  
$A_n$ can be calculated by following recursive formula, obtained using (\ref{4.9}) and (\ref{4.12}), 
\begin{equation}\label{4.14}
A_n(u_0,u_1,\ldots,u_n)=\frac{1}{2n\pi}\int^\pi_{-\pi} A_{n-1}(v_0,v_1,\ldots,v_{n-1})e^{-i\lambda} \,d\lambda,\ \forall\ n\in \mathbb{N}, 
\end{equation}
where\ $v_k=u_k+(k+1)u_{k+1}e^{i\lambda}$ and $\overline{v}_k=\overline{u}_k+(k+1)\overline{u}_{k+1}e^{i\lambda}$, $k\in\{0,1,2,\ldots,n\}$.

\begin{center}
{\bf 5. APPLICATIONS TO SOME FORMS OF NONLINEARITY}
\end{center}

In this section, we apply the above discussed methods to calculate Adomian polynomials for different forms of nonlinearity. The second method is efficient in cases where Taylor series expansion is required, as in case of exponential, logarithmic and trigonometric nonlinearity. The advantage is that the second algorithm requires at most the first two terms of the Taylor series expansion. Applications of properties of Adomian polynomials discussed in the third section are also illustrated.
\vspace{.15cm}\\
\noindent {\small{\sc Example 1.}} (Nonlinear polynomials) Adomian polynomials for $N(u)=\overline{u}u^m$, where $m\in\mathbb{N}$. We use (\ref{4.8}) to find $A_n$. Obviously, $A_0=|u_0|^2u_0^{m-1}$ and
\begin{align*} 
A_1&=\frac{1}{2\pi}\int^\pi_{-\pi} (\overline{u}_0 + \overline{u}_1e^{i\lambda})(u_0 + u_1e^{i\lambda})^m e^{-i\lambda} \,d\lambda\\
&=\frac{1}{2\pi}\int^\pi_{-\pi} (\overline{u}_0 + \overline{u}_1e^{i\lambda})\sum_{k=0}^{m}{\binom{m}{k}}u_0^k(u_1e^{i\lambda})^{m-k}e^{-i\lambda}  \,d\lambda\\
&=u_0^m\overline{u}_1+mu_0^{m-2}|{u}_0|^2u_1,\\
A_2&=\frac{1}{2\pi}\int^\pi_{-\pi} (\overline{u}_0 + \overline{u}_1e^{i\lambda}+ \overline{u}_2e^{2i\lambda})(u_0 + u_1e^{i\lambda} + u_2e^{2i\lambda} )^me^{-2i\lambda}  \,d\lambda\\
&=\frac{1}{2\pi}\int^\pi_{-\pi} (\overline{u}_0 + \overline{u}_1e^{i\lambda}+ \overline{u}_2e^{2i\lambda})\\
&\ \ \ \sum_{\sum_{j=1}^3k_j=m}{\binom{m}{k_1,k_2,k_3}}u_0^{k_1}(u_1e^{i\lambda})^{k_2}(u_2e^{2i\lambda})^{k_3}e^{-2i\lambda} \,d\lambda\\
&=u_0^m\overline{u}_2+mu_0^{m-1}|u_1|^2+mu_0^{m-2}|u_0|^2u_2+\frac{1}{2}m(m-1)u_0^{m-3}|u_0|^2u_1^2.
\end{align*}
Indeed, from Theorem \ref{t1} (ii), the $n$-th order Adomian polynomial is given by
\begin{equation*}\label{5.1}
A_n(u_0,u_1,\ldots,u_n)=\underset{k_j\in\mathbb{N}_0}{\sum_{\sum_{j=1}^{m+1}{k_j}=n}}\ \overline{u}_{k_{m+1}}\prod_{j=1}^{m} u_{k_j},\ \forall\ \ n\in \mathbb{N}_0.
\end{equation*}

\noindent {\small{\sc Example 2.}} (Trigonometric function) Adomian polynomials for $N(u)= \sin{u}$. Using (\ref{4.8}), $A_0=\sin{u_0}$ and
\begin{align*}
A_1&=\frac{1}{2\pi}\int^\pi_{-\pi} \left[\sin({u_0+u_1e^{i\lambda}})\right]e^{-i\lambda}\,d\lambda\\
&=\frac{1}{2\pi}\int^\pi_{-\pi} \left[\cos{(u_1e^{i\lambda})}\sin{u_0} +\sin{(u_1e^{i\lambda})}\cos{u_0}\right]e^{-i\lambda} \,d\lambda\\
&=\frac{1}{2\pi}\int^\pi_{-\pi} \left[\left\{1-\frac{u_1^2e^{2i\lambda}}{2!}+\ldots\right\}\sin{u_0}+ \left\{u_1e^{i\lambda}-\frac{u_1^3e^{3i\lambda}}{3!}+\ldots\right\}\cos{u_0}\right]e^{-i\lambda}\,d\lambda\\
&=u_1\cos{u_0},\\
A_2&=\frac{1}{2\pi}\int^\pi_{-\pi} \left[\sin({u_0+u_1e^{i\lambda}+u_2e^{2i\lambda}})\right]e^{-2i\lambda}\,d\lambda\\
&=\frac{1}{2\pi}\int^\pi_{-\pi} \left[\left\{1-\frac{(u_1e^{i\lambda}+u_2e^{2i\lambda})^2}{2!}+\ldots\right\}\sin{u_0}\right.\\
&\ \ \ + \left\{(u_1e^{i\lambda}+u_2e^{2i\lambda})-\ldots\right\}\cos{u_0}\bigg]e^{-2i\lambda}\,d\lambda\\
&=u_2\cos{u_0}-\frac{1}{2}u_1^2\sin{u_0},\\
A_3&=\frac{1}{2\pi}\int^\pi_{-\pi} \left[\sin({u_0+u_1e^{i\lambda}+u_2e^{2i\lambda}+u_3e^{3i\lambda}})\right]e^{-3i\lambda}\,d\lambda\\
&=\frac{1}{2\pi}\int^\pi_{-\pi} \left[\left\{1-\frac{(u_1e^{i\lambda}+u_2e^{2i\lambda}+u_3e^{3i\lambda})^2}{2!}+\ldots\right\}\sin{u_0}\right.+\bigg\{(u_1e^{i\lambda}+u_2e^{2i\lambda}\\
&\ \ \ +u_3e^{3i\lambda})-\left.\left.\frac{(u_1e^{i\lambda}+u_2e^{2i\lambda}+u_3e^{3i\lambda})^3}{3!}+\ldots\right\} \cos{u_0}\right]e^{-3i\lambda}\,d\lambda\\
&=u_3\cos{u_0}-\frac{1}{6}u_1^3\cos{u_0}-u_1u_2\sin{u_0}.
\end{align*} 
Similarly, $A_4, A_5,\ldots$ can be calculated.
\vspace{.15cm}\\
\noindent {\small{\sc Example 3.}} (Exponential function) Adomian polynomials for $N(u)= e^u$. From (\ref{4.8}), we have $A_0= e^{u_0}$ and
\begin{align*} 
A_1&=\frac{1}{2\pi}\int^\pi_{-\pi} e^{u_0 + u_1e^{i\lambda}} e^{-i\lambda}\,d\lambda\\
&=\frac{1}{2\pi}\int^\pi_{-\pi} e^{u_0}\left\{1+ \frac{u_1e^{i\lambda}}{1!} +\ldots\right\} e^{-i\lambda}\,d\lambda\\
&=u_1e^{u_0},\\
A_2&=\frac{1}{2\pi}\int^\pi_{-\pi} e^{u_0 + u_1e^{i\lambda} + u_2e^{2i\lambda}} e^{-2i\lambda}\,d\lambda\\
&=\frac{1}{2\pi}\int^\pi_{-\pi} e^{u_0}\left\{1+ \frac{(u_1e^{i\lambda}+ u_2e^{2i\lambda})}{1!} +\frac{(u_1e^{i\lambda}+ u_2e^{2i\lambda})^2}{2!} + \ldots\right\} e^{-2i\lambda}\,d\lambda\\
&=\left(u_2 + \frac{u_1^2}{2}\right)e^{u_0},
\end{align*}
\noindent and etc. Indeed, from Theorem \ref{t1} (v), the $n$-th order Adomian polynomial for $e^u$ is
\begin{equation*}\label{5.11}
A_n(u_0,u_1,\ldots,u_n)=e^{u_0}\underset{k_j\in\mathbb{N}_0}{\sum_{\sum_{j=1}^{n}jk_j=n}}\prod_{j=1}^{n}\frac{u_j^{k_j}}{k_j!},\ \forall\ n\in \mathbb{N}.
\end{equation*}

\noindent {\small{\sc Example 4.}} (Composite nonlinearity) Adomian polynomials for $N(u)= e^{\sin{u}}$. Using (\ref{4.8}), $A_0=e^{\sin{u_0}}$ and
\begin{align*} 
A_1&=\frac{1}{2\pi}\int^\pi_{-\pi} e^{\sin{(u_0+u_1e^{i\lambda})}}e^{-i\lambda} \,d\lambda\\
&=\frac{e^{\sin{u_0}}}{2\pi}\int^\pi_{-\pi} e^{\left(\sin{u_1e^{i\lambda}}\cos{u_0}-2\sin^2{\frac{u_1e^{i\lambda}}{2}}\sin{u_0}\right)}e^{-i\lambda} \,d\lambda\\
&=\frac{e^{\sin{u_0}}}{2\pi}\int^\pi_{-\pi} \bigg[1+\bigg(\left\{u_1e^{i\lambda}-\ldots\right\}\cos{u_0}\\
&\ \ \ -2\left\{\frac{1}{2}u_1e^{i\lambda}-\ldots\right\}^2\sin{u_0}\bigg)+\ldots\bigg]e^{-i\lambda}\,d\lambda\\
&=u_1\cos{u_0}e^{\sin{u_0}},\\
A_2&=\frac{1}{2\pi}\int^\pi_{-\pi} e^{\sin{(u_0+u_1e^{i\lambda}+u_2e^{2i\lambda})}} e^{-2i\lambda}\,d\lambda\\
&=\frac{e^{\sin{u_0}}}{2\pi}\int^\pi_{-\pi} \left[1+\Bigg(\left\{(u_1e^{i\lambda}+u_2e^{2i\lambda})-\ldots\right\}\cos{u_0}\right.\\
&\ \ \ \left.-2\left\{\frac{1}{2}(u_1e^{i\lambda}+u_2e^{2i\lambda})-\ldots\right\}^2\sin{u_0}\right)+\frac{1}{2!}\Bigg(\left\{(u_1e^{i\lambda}+u_2e^{2i\lambda})-\ldots\right\}\cos{u_0}\\
&\ \ \ \left.\left.-2\left\{\frac{1}{2}(u_1e^{i\lambda}+u_2e^{2i\lambda})-\ldots\right\}^2\sin{u_0}\right)^2+\ldots\right]e^{-2i\lambda}\,d\lambda\\
&=\left(u_2\cos{u_0}-\frac{1}{2}u_1^2\sin{u_0}+\frac{1}{2}u_1^2\cos^2{u_0}\right)e^{\sin{u_0}}.
\end{align*}
Using Theorem \ref{t1} (v), we obtain
\begin{equation*}\label{5.13}
A_n(u_0,u_1,\ldots,u_n)=e^{\sin{u_0}}\underset{k_j\in\mathbb{N}_0}{\sum_{\sum_{j=1}^{n}jk_j=n}}\prod_{j=1}^{n}\frac{B_j^{k_j}}{k_j!},\ \forall\ n\in \mathbb{N},
\end{equation*}
\noindent where $B_n$ are Adomian polynomials of $\sin{u}$.\\

\noindent The next example is based on the second method.
\vspace{.15cm}\\
\noindent {\small{\sc Example 5.}} (Logarithmic function) Adomian polynomials for $N(u)= \ln{u}$. Obviously, $A_0=\ln{u_0}$, from (\ref{4.13}). Also, from (\ref{4.14}),
\begin{align*}
A_1&=\frac{1}{2\pi}\int^\pi_{-\pi} \ln({u_0 + u_1e^{i\lambda}}) e^{-i\lambda} \,d\lambda\\
&=\frac{1}{2\pi}\int^\pi_{-\pi} \ln{\left[u_0\left(1 +\frac{u_1e^{i\lambda}}{u_0}\right)\right]}e^{-i\lambda} \,d\lambda\\
&=\frac{1}{2\pi}\int^\pi_{-\pi} \left[\ln{u_0}+\left\{\frac{u_1e^{i\lambda}}{u_0}+\ldots \right\}\right] e^{-i\lambda} \,d\lambda\\ 
&=\frac{u_1}{u_0},\\
A_2&=\frac{1}{4\pi}\int^\pi_{-\pi} \frac{(u_1+2u_2e^{i\lambda})}{(u_0+u_1e^{i\lambda})}e^{-i\lambda}\,d\lambda\\
&=\frac{1}{4\pi}\int^\pi_{-\pi} \frac{(u_1+2u_2e^{i\lambda})}{u_0}\left(1+\frac{u_1e^{i\lambda}}{u_0}\right)^{-1}e^{-i\lambda} \,d\lambda\\
&=\frac{1}{4\pi}\int^\pi_{-\pi} \frac{(u_1+2u_2e^{i\lambda})}{u_0}\left\{1-\frac{u_1e^{i\lambda}}{u_0}+\ldots\right\}e^{-i\lambda}\,d\lambda\\
&=\frac{u_2}{u_0} -\frac{u_1^2}{2u_0^2},
\end{align*}
and etc. Indeed, from Theorem \ref{t1} (v), we get
\begin{equation*}\label{5.12}
A_n(u_0,u_1,\ldots,u_n)=\underset{k_j\in\mathbb{N}_0}{\sum_{\sum_{j=1}^{n}jk_j=n}}\frac{(-1)^{\sum_{j=1}^{n}k_j-1}\left(\sum_{j=1}^{n}k_j-1\right)!}{u_0^{\sum_{j=1}^{n}k_j}}\prod_{j=1}^{n}\frac{u_j^{k_j}}{k_j!},\ \forall\ n\in \mathbb{N}.
\end{equation*}

\begin{center}
{\bf 6. EXTENSION TO THE CASE OF SEVERAL VARIABLES}
\end{center}

We here extend our methods to calculate Adomian polynomials for the multivariable case. Consider the system of $m$ functional equations,
\begin{equation}\label{6.1}
u_j=f_j+L_j(u_1,u_2,\ldots,u_m)+N_j(u_1,u_2,\ldots,u_m),\ j=1,2,\ldots,m.
\end{equation}
Here, $L_j$'s and $N_j$'s are linear and nonlinear operators respectively and $f_j$'s are known functions. As assumed earlier, we shall suppose
\begin{itemize}
\item[$H3:$] Solution $u_j=\sum_{k=0}^{\infty}u_{j_k}$ of (\ref{6.1}) are absolutely convergent for $j=1,2,\ldots,m$.
\item[$H4:$]  The nonlinear function $N_j(u_1,u_2,\ldots,u_m)$ is developable into an entire series with infinite radius of convergence so that
\end{itemize} 
\begin{equation}\label{6.2}
N_j(u_1,u_2,\ldots,u_m)=\sum_{k_1=0}^{\infty}\sum_{k_2=0}^{\infty}\ldots\sum_{k_m=0}^{\infty}\frac{\partial^{k_1+k_2\ldots+k_m}N_j(0,0,\ldots,0)}{\partial^{k_1}u_1\partial^{k_2}u_2\ldots\partial^{k_m}u_m}\prod_{j=1}^{m}\frac{u_j^{k_j}}{k_j!},
\end{equation}
for all $1\leq j\leq m$. Since (\ref{6.2}) is absolutely convergent, it can be rearranged as
\begin{align}\label{6.3}
N_j(u_1,u_2,\ldots,u_m)&=‎‎\sum_{k=0}^{\infty‎}A_{j_k}(u_{1_0},\ldots,u_{1_k},u_{2_0},\ldots,u_{2_k},\ldots,u_{m_0},\ldots,u_{m_k})\ \nonumber \\
&=\sum_{k_1=0}^{\infty}\ldots\sum_{k_m=0}^{\infty}\frac{\partial^{k_1+\ldots+k_m}N_j(u_{1_0},\ldots,u_{m_0})}{\partial^{k_1}u_1\ldots\partial^{k_m}u_m}\prod_{j=1}^{m}\frac{(u_j-u_{j_0})^{k_j}}{k_j!}.
\end{align}
Parameterize $u_j(x,t)$ and its complex conjugate $\overline{u}_j(x,t)$ as follows:
\begin{equation*}\label{6.4}
u_{j_\lambda}=‎‎\sum_{k=0}^{\infty‎}u_{j_k}f^k(\lambda),\ \ \ \ \overline{u}_{j_\lambda}=‎‎\sum_{k=0}^{\infty‎}\overline{u}_{j_\lambda}f^k(\lambda)\ ,\ \ 1\leq j \leq m,
\end{equation*}
\noindent where $\lambda$ is a real parameter and $f$ is any real or complex valued function with $|f|<1$. Since the series (\ref{6.3}) is absolutely convergent, $N_j(u_{1_\lambda},u_{2_\lambda},\ldots,u_{m_\lambda})$ can be decomposed as
\begin{equation}\label{6.5}
N_j(u_{1_\lambda},u_{2_\lambda},\ldots,u_{m_\lambda})=‎‎\sum_{k=0}^{\infty‎}A_{j_k}(u_{1_0},\ldots,u_{1_k},\ldots,u_{m_0},\ldots,u_{m_k})f^k(\lambda).
\end{equation}
Taking $f(\lambda)=e^{i\lambda}$, the parametrized form of $u_j(x,t)$, for each $j$, is
\begin{equation}\label{6.6}
u_{j_\lambda}=‎‎\sum_{k=0}^{\infty‎}u_{j_k} e^{ik\lambda}
\end{equation}
and complex conjugates, $\overline{u}_j(x,t)$ are parametrized as $ \overline{u}_{j_\lambda}=‎‎\sum_{k=0}^{\infty‎}\overline{u}_{j_k} e^{ik\lambda}$. We first give the extended version of Theorem \ref{t2} for the multivariable case.
\begin{theorem}\label{t3}
Let the parametrized representation of $u_j(x,t)$, $1\leq j \leq m$, be given by (\ref{6.6}), where $\lambda$ is a real parameter and $N_j(u_1,u_2,\ldots,u_m)$ are the nonlinear terms in (\ref{6.1}). Then
\begin{align*}\label{6.7}
\int^\pi_{-\pi} N_j(u_{1_\lambda},u_{2_\lambda}&,\ldots,u_{m_\lambda}) e^{-in\lambda} \,d\lambda\nonumber\\&=\int^\pi_{-\pi} N_j\left(‎‎\sum_{k=0}^{n‎}u_{1_k} e^{ik\lambda},‎‎\sum_{k=0}^{n‎}u_{2_k} e^{ik\lambda},\ldots,‎‎\sum_{k=0}^{n‎}u_{m_k} e^{ik\lambda}\right) e^{-in\lambda} \,d\lambda.
\end{align*}
\end{theorem}
\begin{proof}
Let us consider $m$-tuple vectors $\alpha=(\alpha_1,\alpha_2,\ldots,\alpha_m)$, $\textbf{u}=(u_1,u_2,\ldots,u_m)$, $\textbf{u}_0=(u_{1_0},u_{2_0},\ldots,u_{m_0})$, $\textbf{u}_\lambda=\left(‎‎\sum_{k=0}^{\infty‎}u_{1_k} e^{ik\lambda},‎‎\sum_{k=0}^{\infty‎}u_{2_k} e^{ik\lambda},\ldots,‎‎\sum_{k=0}^{\infty‎}u_{m_k} e^{ik\lambda}\right)$\\ and $\textbf{u}_{n_\lambda}=\left(‎‎\sum_{k=0}^{n‎}u_{1_k} e^{ik\lambda},‎‎\sum_{k=0}^{n‎}u_{2_k} e^{ik\lambda},\ldots,\right.$ $\left.‎‎\sum_{k=0}^{n‎}u_{m_k} e^{ik\lambda}\right)$.\\ Also, denote $|\alpha|=\sum_{k=1}^m\alpha_k$, $\textbf{u}^\alpha=\prod_{k=1}^m u_k^{\alpha_k}$, $\alpha!=\prod_{k=1}^m\alpha_k!$ and $\partial^\alpha=\prod_{k=1}^m\frac{\partial^{\alpha_k}}{\partial u^{\alpha_k}_k}$. From $H3$, $\sum_{k=1}^{\infty‎}|u_{j_k}|=M_j<\infty$ for $j=1,2,\ldots,m$ and therefore\\ 
\begin{equation*}\label{6.8}
\left|\partial^\alpha N_j(\textbf{u}_0)\frac{(\textbf{u}_\lambda-\textbf{u}_0)^\alpha}{\alpha!}\right|\leq
\left|\frac{\partial^\alpha N_j(\textbf{u}_0)}{\alpha!}\right|\textbf{M}^{\alpha},
\end{equation*}
where $\textbf{M}=(M_1,M_2,\ldots,M_m)$. Using $H4$, $\sum_{|\alpha|\geq0}\left|\frac{\partial^\alpha N_j(\textbf{u}_0)}{\alpha!}\right|\textbf{M}^{\alpha}<\infty$. Hence, by Weierstrass M-test, $$\sum_{|\alpha|\geq0}\partial^\alpha N_j(\textbf{u}_0)\frac{(\textbf{u}_\lambda-\textbf{u}_0)^\alpha}{\alpha!}$$ converges uniformly. Hence, for $n\in\mathbb{N}_0$ and using (\ref{6.3}), we get
\begin{align*}\label{new2}
\int^\pi_{-\pi} N_j(\textbf{u}_\lambda) e^{-in\lambda} \,d\lambda
&=\int^\pi_{-\pi}\sum_{|\alpha|\geq0}\partial^\alpha N_j(\textbf{u}_0)\frac{(\textbf{u}_\lambda-\textbf{u}_0)^\alpha}{\alpha!} e^{-in\lambda} \,d\lambda\nonumber\\
&=\underset{m\rightarrow \infty}{\lim}\sum_{|\alpha|=m}\int^\pi_{-\pi}\partial^\alpha N_j(\textbf{u}_0)\frac{(\textbf{u}_{n_\lambda}-\textbf{u}_0)^\alpha}{\alpha!}  e^{-in\lambda} \,d\lambda,\ \ \mathrm{(by\ (\ref{4.2}))}\\
&=\int^\pi_{-\pi}\sum_{|\alpha|\geq0}\partial^\alpha N_j(\textbf{u}_0)\frac{(\textbf{u}_{n_\lambda}-\textbf{u}_0)^\alpha}{\alpha!} e^{-in\lambda} \,d\lambda\\
&=\int^\pi_{-\pi} N_j(\textbf{u}_{n_\lambda}) e^{-in\lambda} \,d\lambda,
\end{align*}
and thus the proof is complete using (\ref{6.3}).
\end{proof}

\subsection*{6.1. Extension of the First Method}

\hspace{.75cm} Taking $f(\lambda)=e^{i\lambda}$, we have from (\ref{6.5}),
\begin{equation}\label{6.9}
N_j(u_{1_\lambda},u_{2_\lambda},\ldots,u_{m_\lambda})=‎‎\sum_{k=0}^{\infty‎}A_{j_k}(u_{1_0},\ldots,u_{1_k},\ldots,u_{m_0},\ldots,u_{m_k}) e^{ik\lambda},
\end{equation}
for $j=1,2,\ldots,m$. To determine $A_{j_n}$, multiply $e^{-in\lambda}$ in (\ref{6.9}) and integrate both sides with respect to $\lambda$ from $-\pi$ to $\pi$, to get
\begin{align}\label{6.10}
\int^\pi_{-\pi} N_j\bigg(‎‎\sum_{k=0}^{\infty‎}u_{1_k} e^{ik\lambda},‎‎\sum_{k=0}^{\infty‎}u_{2_k} e^{ik\lambda}&,\ldots,‎‎\sum_{k=0}^{\infty‎}u_{m_k} e^{ik\lambda}\bigg) e^{-in\lambda}\,d\lambda\nonumber\\&=\int^\pi_{-\pi} ‎‎\sum_{k=0}^{\infty‎}A_{j_k} e^{ik\lambda} e^{-in\lambda} \,d\lambda=2\pi A_{j_n}.
\end{align}
The last equality in (\ref{6.10}) follows due to the uniform convergence of the series $\sum_{k=0}^{\infty‎}A_{j_k} e^{i(k-n)\lambda}$. Hence,
\begin{equation*}\label{6.11}
A_{j_n}(u_{1_0},\ldots,u_{m_n})=\frac{1}{2\pi}\int^\pi_{-\pi} N_j\left(‎‎\sum_{k=0}^{\infty‎}u_{1_k} e^{ik\lambda},\ldots,‎‎\sum_{k=0}^{\infty‎}u_{m_k} e^{ik\lambda}\right) e^{-in\lambda} \,d\lambda.
\end{equation*}
Applying Theorem \ref{t3}, we get for $j=1,2,\ldots,m$ and $n\in \mathbb{N}_0$,
\begin{equation}\label{6.12}
A_{j_n}(u_{1_0},\ldots,u_{m_n})=\frac{1}{2\pi}\int^\pi_{-\pi} N_j\left(‎‎\sum_{k=0}^{n‎}u_{1_k} e^{ik\lambda},\ldots,‎‎\sum_{k=0}^{n‎}u_{m_k} e^{ik\lambda}\right) e^{-in\lambda} \,d\lambda. 
\end{equation}

\noindent {\small{\sc Example 6.}} Consider the set of nonlinear equations
\begin{equation*}
N_j(u_1,u_2,u_3)=u_1\frac{\partial u_j}{\partial x}+u_2\frac{\partial u_j}{\partial y}+u_3\frac{\partial u_j}{\partial z},\ j=1,2,3.
\end{equation*}

These nonlinear terms appear in the Navier Stokes equation for an incompressible fluid flow defined by
\begin{equation}\label{6.16}
\frac{\partial V}{\partial t}+(V.\nabla)V=\frac{\eta}{\rho}\Delta v-\frac{1}{\rho}\nabla p.
\end{equation}
\noindent Here $x,y,z$ are spatial components, $t$ is the temporal component, $\eta$ denotes dynamic viscosity, $\rho$ denotes density, $\nu=\eta/\rho$ is the kinematic viscosity and $V=(u_1,u_2,u_3)$ denotes the speed vector. Using prevalent methods, {\small{\sc Seng}} et al. \cite{Seng199659} computed the Adomian polynomials for the nonlinear term $(V.\nabla)V$ in (\ref{6.16}), using tedious computations.

By our extended first method, we calculate $A_n$ with a few steps. From (\ref{6.12}), Adomian polynomials $A_{j_n}$ for $j=1,2,3$ are
\begin{align*}
A_{j_0}&=u_{1_0}\frac{\partial u_{j_0}}{\partial x}+u_{2_0}\frac{\partial u_{j_0}}{\partial y}+u_{3_0}\frac{\partial u_{j_0}}{\partial z},\\
A_{j_1}&=\frac{1}{2\pi}\int^\pi_{-\pi} \left[(u_{1_0}+u_{1_1}e^{i\lambda})\frac{\partial (u_{j_0}+u_{j_1}e^{i\lambda})}{\partial x}+(u_{2_0}+u_{2_1}e^{i\lambda})\frac{\partial (u_{j_0}+u_{j_1}e^{i\lambda})}{\partial y}\right.\\
&\ \ \ \left.+(u_{3_0}+u_{3_1}e^{i\lambda})\frac{\partial (u_{j_0}+u_{j_1}e^{i\lambda})}{\partial z} \right]e^{-i\lambda}\,d\lambda\\
&=u_{1_0}\frac{\partial u_{j_1}}{\partial x}+u_{1_1}\frac{\partial u_{j_0}}{\partial x}+u_{2_0}\frac{\partial u_{j_1}}{\partial y}+u_{2_1}\frac{\partial u_{j_0}}{\partial y}+u_{3_0}\frac{\partial u_{j_1}}{\partial z}+u_{3_1}\frac{\partial u_{j_0}}{\partial z},\\
A_{j_2}&=\frac{1}{2\pi}\int^\pi_{-\pi} \left[(u_{1_0}+u_{1_1}e^{i\lambda}+u_{1_2}e^{2i\lambda})\frac{\partial (u_{j_0}+u_{j_1}e^{i\lambda}+u_{j_2}e^{2i\lambda})}{\partial x}\right.\\
&\ \ \ \left.+(u_{2_0}+u_{2_1}e^{i\lambda}+u_{2_2}e^{2i\lambda})\frac{\partial (u_{j_0}+u_{j_1}e^{i\lambda}+u_{j_2}e^{2i\lambda})}{\partial y}\right.\\
&\ \ \ \left.+(u_{3_0}+u_{3_1}e^{i\lambda}+u_{3_2}e^{2i\lambda})\frac{\partial (u_{j_0}+u_{j_1}e^{i\lambda}+u_{j_2}e^{2i\lambda})}{\partial z} \right]e^{-2i\lambda} \,d\lambda\\
&=u_{1_0}\frac{\partial u_{j_2}}{\partial x}+u_{1_1}\frac{\partial u_{j_1}}{\partial x}+u_{1_2}\frac{\partial u_{j_0}}{\partial x}+u_{2_0}\frac{\partial u_{j_2}}{\partial y}\\
&\ \ \ +u_{2_1}\frac{\partial u_{j_1}}{\partial y}+u_{2_2}\frac{\partial u_{j_0}}{\partial y}+u_{3_0}\frac{\partial u_{j_2}}{\partial z}+u_{3_1}\frac{\partial u_{j_1}}{\partial z}+u_{3_2}\frac{\partial u_{j_0}}{\partial z}.
\end{align*}
Thus, the $n$-th order Adomian polynomials for $j=1,2,3$ are given by
\begin{equation*}
A_{j_n}(u_{1_0},\ldots,u_{1_n},\ldots,u_{3_0},\ldots,u_{3_n})=\sum_{(k,w)\in\{(1,x),(2,y),(3,z)\}}‎‎\underset{a,b\in\mathbb{N}_0}{\sum_{a+b=n}}u_{k_a}\frac{\partial u_{j_b}}{\partial w},\ \forall\ n\in\mathbb{N}_0.
\end{equation*}

\subsection*{6.2. Extension of the Second Method}

\hspace{.75cm} The Adomian polynomials can be calculated recursively for the multivariable case also. {\small{\sc Duan}} \cite{Duan20102456} introduced the simplified index matrices of the multivariable Adomian polynomials and established a recurrence relationships among them to provide a convenient recursive algorithm.

Based on our approach, we give a new recursive method to obtain these polynomials for the multivariable case. We define an operator $T$ as
\begin{eqnarray}\label{6.13}
&&T\left(A_{j_n}(u_{1_0},\ldots,u_{1_n},u_{2_0},\ldots,u_{2_n},\ldots,u_{m_0},\ldots,u_{m_n})\right)\\&&~~~~~~~~~~=\frac{1}{2\pi}\int^\pi_{-\pi} A_{j_{n}}(v_{1_0},\ldots,v_{1_{n}},v_{2_0},\ldots,v_{2_{n}},\ldots,v_{m_0},\ldots,v_{m_{n}})e^{-i\lambda} \,d\lambda,\nonumber 
\end{eqnarray}
where\ $v_{j_k}=u_{j_k}+(k+1)u_{j_{k+1}}e^{i\lambda}$ and $\overline{v}_{j_k}=\overline{u}_{j_k}+(k+1)\overline{u}_{j_{k+1}}e^{i\lambda}$, $k\in\{0,1,2,\ldots,n\}$.\\
From (\ref{6.12}), we get for $j=1,2,\ldots,m$,
\begin{equation}\label{6.14}
A_{j_0}(u_{1_0},u_{2_0},\ldots,u_{m_0})= N_j(u_{1_0},u_{2_0},\ldots,u_{m_0}).
\end{equation} 

Note that operator $T$ defined in (\ref{6.13}) satisfies all the properties of Proposition \ref{l4.1}. Therefore, by applying (\ref{4.12}), we get the following recursive formula for $A_{j_n}\ (1\leq j\leq m,\ n\in\mathbb{N})$:
\begin{equation}\label{6.15}
A_{j_n}(u_{1_0},\ldots,u_{m_n})=\frac{1}{2n\pi}\int^\pi_{-\pi} A_{j_{n-1}}(v_{1_0},\ldots,v_{m_{n-1}})e^{-i\lambda} \,d\lambda, 
\end{equation}
where $v_{j_k}=u_{j_k}+(k+1)u_{j_{k+1}}e^{i\lambda}$ and $\overline{v}_{j_k}=\overline{u}_{j_k}+(k+1)\overline{u}_{j_{k+1}}e^{i\lambda}$, $k\in\{0,1,\ldots,n-1\}$.\\

\noindent {\small{\sc Example 7.}} Adomian polynomials for $N(u)=\overline{u}_1u_2\frac{\partial u_2}{\partial x}$. From (\ref{6.14}), we have $A_0=\overline{u}_{1_0}u_{2_0}\frac{\partial u_{2_0}}{\partial x}$. Now by using (\ref{6.15}),
\begin{align*}
A_1&=\frac{1}{2\pi}\int^\pi_{-\pi}(\overline{u}_{1_0} + \overline{u}_{1_1}e^{i\lambda})(u_{2_0} + u_{2_1}e^{i\lambda})\frac{\partial (u_{2_0} + u_{2_1}e^{i\lambda})}{\partial x} e^{-i\lambda} \,d\lambda\\
&=\overline{u}_{1_0}u_{2_0}\frac{\partial u_{2_1}}{\partial x}+\overline{u}_{1_0}u_{2_1}\frac{\partial u_{2_0}}{\partial x}+\overline{u}_{1_1}u_{2_0}\frac{\partial u_{2_0}}{\partial x},\\
A_2&=\frac{1}{4\pi}\int^\pi_{-\pi}\left[(\overline{u}_{1_0} + \overline{u}_{1_1}e^{i\lambda})(u_{2_0} + u_{2_1}e^{i\lambda})\frac{\partial (u_{2_1} + 2u_{2_2}e^{i\lambda})}{\partial x}
\right.\\
&\ \ \ +(\overline{u}_{1_0} + \overline{u}_{1_1}e^{i\lambda})(u_{2_1} + 2u_{2_2}e^{i\lambda})\frac{\partial (u_{2_0} + u_{2_1}e^{i\lambda})}{\partial x}\\
&\ \ \ \left.+(\overline{u}_{1_1} + 2\overline{u}_{1_2}e^{i\lambda})(u_{2_0} + u_{2_1}e^{i\lambda})\frac{\partial (u_{2_0} + u_{2_1}e^{i\lambda})}{\partial x}\right]e^{-i\lambda} \,d\lambda\\
&=\overline{u}_{1_0}u_{2_0}\frac{\partial u_{2_2}}{\partial x}+\overline{u}_{1_0}u_{2_1}\frac{\partial u_{2_1}}{\partial x}+\overline{u}_{1_1}u_{2_0}\frac{\partial u_{2_1}}{\partial x}\\
&\ \ \ +\overline{u}_{1_0}u_{2_2}\frac{\partial u_{2_0}}{\partial x}+\overline{u}_{1_1}u_{2_1}\frac{\partial u_{2_0}}{\partial x}+\overline{u}_{1_2}u_{2_0}\frac{\partial u_{2_0}}{\partial x}.
\end{align*}
Thus, the $n$-th order Adomian polynomials are given by
\begin{equation*}\label{6.17}
A_n(u_{1_0},\ldots,u_{1_n},u_{2_0},\ldots,u_{2_n}) =‎‎\underset{a,b\in\mathbb{N}_0}{\sum_{a+b+c=n}}\overline{u}_{1_a}u_{2_b}\frac{\partial u_{2_c}}{\partial x},\ \forall\ n\in\mathbb{N}_0.
\end{equation*}

\noindent {\small{\sc Remark 3.}}
The recursive algorithm obtained here is based on a simple integration, which on using (\ref{4.2}) conveniently produces Adomian polynomials. More analytic recursive algorithms based on regular operations such as addition, multiplication and differentiation were obtained by {\small{\sc Duan}} \cite{Duan20116337}.

\begin{center}
{\bf 7. CONCLUSIONS}
\end{center}

The crucial step involved in Adomian decomposition method is the employment of the ``Adomian polynomials". The computation of $n$-th order Adomian polynomial is difficult, as it requires tedious calculations. In this paper, we have discussed some important properties of Adomian polynomials and developed two simple methods which avoid draggy calculation of higher derivatives involved in prevalent methods. Another advantage is that at every stage we don't have to keep track of sum of the indices of components of $u(x,t)$ (see {\small{\sc Wazwaz}} \cite{Wazwaz200033}). Also, the second algorithm is efficient in cases where Taylor series expansion is required, as for example in case of exponential, logarithmic and trigonometric nonlinearity, and it just requires the first two terms of the Taylor series expansion. We have illustrated our approach using typical examples.\\

\noindent \textbf{Acknowledgements.}
The authors wish to thank two anonymous referees for providing insightful comments and suggestions. The research of the first author was supported by UGC, Govt. of India grant F.2-2/98(SA-I).

%
%

\vspace{1cc}


{\small
\noindent
Department of Mathematics,\\
Indian Institute of Technology Bombay,\\
Powai, Mumbai, MH 400076\\
India\\
E-mails: kulkat@math.iitb.ac.in\\
\hspace*{1.2cm} pv@math.iitb.ac.in

}\end{document}